\documentclass[smallextended]{svjour3}

\usepackage{graphics} 
\usepackage{epsfig} 
\usepackage{mathptmx} 
\usepackage{times} 
\usepackage{amsmath} 
\usepackage{amssymb}  
\usepackage{cite}
\usepackage{multirow}
\usepackage{mathtools}

\usepackage{amsthm}
\usepackage{slashbox}
\usepackage{epstopdf}

\newtheorem{algorithm}[theorem]{Algorithm}
\newtheorem{assumption}[theorem]{Assumption}

\DeclareMathOperator*{\argmax}{arg\,max}
\DeclareMathOperator*{\argmin}{arg\,min}

\title{Iterative Computation of Security Strategies of Matrix Games with Growing Action Set}

\author{Lichun Li and Cedric Langbort \thanks{This work was supported in part by NSF grants 1619339 and 1151076 to Cedric Langbort.}}
\institute{The authors are with the coordinated science lab at University of Illinois at Urbana-Champaign, Champaign, IL 61801, US. Email: lichunli,langbort@illinois.edu.}

\begin{document}

\maketitle

\begin{abstract}
This paper studies how to efficiently update the saddle-point strategy, or security strategy of one player in a matrix game when the other player develops new actions in the game. It is well known that the saddle-point strategy of one player can be computed by solving a linear program. Developing a new action will add a new constraint to the existing LP. Therefore, our problem becomes how to solve the new LP with a new constraint efficiently. Considering the potentially huge number of constraints, which corresponds to the large size of the other player's action set, we use shadow vertex simplex method, whose computational complexity is  lower than linear with respect to the size of the constraints, as the basis of our iterative algorithm. We first rebuild the main theorems in shadow vertex method with relaxed assumption to make sure such method works well in our model, then analyze the probability that the old optimum remains optimal in the new LP, and finally provides the iterative shadow vertex method whose computational complexity is shown to be strictly less than that of shadow vertex method. The simulation results demonstrates our main results about the probability of re-computing the optimum and the computational complexity of the iterative shadow vertex method.
\end{abstract}

\section{Introduction}
In many non-cooperative games, players can develop new actions as the game is being played. For example, in cyber-security scenarios, attackers can suddenly start employing theretofore unknown system vulnerabilities in the form of so-called 'zero-day attacks'. The goal, then, for the defender, is to quickly and efficiently modify its response so as to reach an equilibrium of the newly formed game.

This paper focuses on two player zero sum games where one player, say player 2, can develop new actions, and has a growing action set. Player 1, who has a fixed action set, may need to update its strategy to the new saddle-point strategy when new action is revealed. Player 1's saddle-point strategy can be computed by solving a linear program (LP), and adding a new action is nothing more than adding a new constraint to the existing LP. Hence, our problem becomes how to solve the LP with a new constraint efficiently.

The corresponding dual problem looks similar to an online LP problem where new variables are introduced to the existing LP \cite{devanur2009adwords,feldman2010online,kleinberg2005multiple,molinaro2013geometry,agrawal2014dynamic}. The key idea to approximate the optimum in online LP is to assign the new variable $1$ if reward is greater than the shadow price based cost, $0$ otherwise. Unlike in the problems considered in this literature, however, the new variables introduced in the LP are not restricted to be boolean, and the central idea developed in these papers cannot be applied to our problem.

Online conic optimization is another online tool which looks for an approximated optimum of a convex optimization problem where new variables are introduced to the existing convex optimization \cite{eghbali2016worst}. A special requirement in online conic optimization is that the constraints of variable are separated. Eghbali et.al provides an algorithm to approximate the optimum based on the form of the optimal primal-dual pair which is related to the concave conjugate of the objective function \cite{eghbali2016worst}. The main issue to adopt the algorithm in our problem is that to separate the constraints of variables, we need to reformulate the objective function, and the concave conjugate of the reformulated objective function is not well defined.

While neither online LP nor online conic optimization can be directly used in our model, our previous work \cite{bopardikar2016incremental} showed how adaptive multiplicative weights could be used to find an approximate solution to the ‘fast and efficient saddle-point strategy updating’ goal as well, by adding a simple condition to the Freund and Schapire scheme. In this paper, we present an algorithm to solve this updating problem exactly.

Considering that our LP may have a large number of constraints (corresponding to a large action set of player 2), we use shadow vertex method, whose computational complexity grows very slowly with respect to the number of constraints, as the basis for our iterative algorithm. Shadow vertex method is a simplex method to solve LP. It first projects the feasible set into a two dimensional plane where the projection of the optimal vertex is still a vertex of the projection of the feasible set, and then walks along the vertices of the projected set to find the optimal one. Because our variables are in the probability space, the non-degeneracy assumption in the original shadow vertex method is violated. Therefore, we first rebuild the main theorems in shadow vertex method with a relaxed non-degeneracy assumption to make sure that shadow vertex methods works well in our model.

Recomputing the optimum is not always necessary when player 2 generates a new action, or in other words, a new constraint is added to the existing LP. If the old optimum satisfies the new constraint, then it remains optimal in the new LP. We further analyze the probability that the old optimum remains optimal in the new LP, and find that if every vertex has the same probability to be the optimal vertex, then the probability that old optimum remains optimal in the new LP increases with respect to the number of constraints. This is because as the number of constraints increases, the feasible set is getting smaller, and it is less possible for the new constraint to cut the optimal vertex off.

If the old optimum is not optimal any more in the new LP, we do not need to start the search all over again. Instead, we can test the feasibility of the previously visited shadow vertices one by one, find out the the feasible one with the best objective value, and use it as the original shadow vertex to start the search. We call the algorithm the iterative shadow vertex method, and show that the computational complexity of iterative shadow vertex method is strictly less than that of regular shadow vertex method.

The rest of the paper is organized as follows. Section II states the problem, and Section III discusses the shadow vertex simplex method with the presence of the probability vector variable. Section IV presents the necessary and sufficient condition of unchanging security strategy and the corresponding possibility, and Section V provides the iterative shadow vertex method and the computational complexity analysis. The simulation results are given in Section VI, followed by the conclusion in Section VII.

\section{Matrix Games with Growing Action Set on One Side}
Let $\mathbb{R}^n$ denote the $n$ dimensional real space. For a finite set $S$, $|S|$ denotes its cardinality, and $\Delta(S)$ is the set of probability space over $S$. Vector $\mathbf{1}$ and $\mathbf{0}$ are column vectors (whose dimension will be clear from context) with all their elements to be one and zero, respectively. $I_n$ is an $n$ dimensional identity matrix and $e_i$ is the $i$th column of $I_n$. Let $u$ and $v$ be two $n$-dimensional vectors. The plane spanned by $u$ and $v$ is denoted by $span(u,v)$, and the angle between $u$ and $v$ is denoted by $arc(u,v)$.

A matrix game is specified by a triple $(S,Q,G)$, where $S$ and $Q$ are finite sets denoting player 1 and 2's actions sets, respectively. The matrix $G\in \mathbb{R}^{|S|\times|Q|}$ is the payoff matrix whose element $G_{s,q}$ is player 1's payoff, or player 2's penalty, if player 1 and 2 play $s\in S$ and $q\in Q$, respectively. In this matrix game, player 1 and player 2 are the maximizer and minimizer, respectively. This paper considers mixed strategy as players' strategy space. Player 1's mixed strategy $\bar{x}\in \Delta(S)$ is a probability over player 1's action set $S$. Player 2's mixed strategy $\bar{y}\in \Delta(Q)$ is defined in the same way. Player 1's expected payoff is $\gamma(\bar{x},\bar{y})= \mathbf{E}_{\bar{x},\bar{y}}(G_{s,q})=\bar{x}^T G \bar{y}.$
Since both players use mixed strategies, there always exists a Nash Equilibrium $(\bar{x}^*,\bar{y}^*)$ such that
\begin{align*}
\gamma(\bar{x},\bar{y}^*)\leq \gamma(\bar{x}^*,\bar{y}^*)\leq \gamma(\bar{x}^*,\bar{y}), \forall \bar{x}\in \Delta(S) \bar{y}\in \Delta(Q),
\end{align*}
When Nash Equilibrium exists, the maxmin value meets the minmax value of the game, i.e. $\max_{\bar{x}\in \Delta(S)}\min_{\bar{y}\in \Delta(Q)}\gamma(\bar{x},\bar{y})=\min_{\bar{y}\in \Delta(Q)}\max_{\bar{x}\in \Delta(S)}\gamma(\bar{x},\bar{y})$. In this case, we say the game has a value $V$, and $\bar{x}^*$, $\bar{y}^*$ are also called the security strategy of player 1 and 2, respectively. Player 1's security strategy can be computed by solving the following linear program \cite{bacsar1998dynamic}.
\begin{align}
  V=&\max_{\bar{x}\in \mathbb{R}^n, \ell \in\mathbb{R}} \ell \label{eq: old LP 1}\\
  s.t.& G^T \bar{x} \geq \ell \mathbf{1}, \label{eq: old LP 2}\\
  & \mathbf{1}^T\bar{x}=1, \label{eq: old LP 3}\\
  &\bar{x}\geq 0. \label{eq: old LP 4}
\end{align}
Player 2's security strategy can be computed by constructing a similar linear program.

This paper considers a special case when player 2's actions are revealed gradually, and the size of player 2's action set is potentially large. Let us suppose that the current size of player 2's action set is $m$, and the size of player 1's action set is $n$ which is fixed. At some point, player 2 develops a new action $q_{new}$, which introduces a new column $g\in \mathbb{R}^n$ to the payoff matrix $G\in \mathbb{R}^{n\times m}$. The change of the payoff matrix may result in the change of player 1's security strategy, and hence player 1 may need to re-run the linear program with extended payoff matrix, as below, to get the new security strategy $\bar{x}^*_{new}$.
\begin{align}
  V_{new}=& \max_{\bar{x}\in\mathbb{R}^n,\ell\in\mathbb{R}} \ell  \label{eq: new LP 1}\\
  s.t. & G^T\bar{x} \geq \ell \mathbf{1} \label{eq: new LP 2}\\
  & g^T\bar{x} \geq \ell\label{eq: new LP 3}\\
  & \mathbf{1}^T \bar{x}=1 \label{eq: new LP 4}\\
  & \bar{x} \geq 0 \label{eq: new LP 5}
\end{align}

Comparing the new LP (\ref{eq: new LP 1}-\ref{eq: new LP 5}) with the old one (\ref{eq: old LP 1}-\ref{eq: old LP 4}), we see that the only difference is the new LP adds one more constraint, equation (\ref{eq: new LP 3}), to the old LP. Our objective is to find an efficient way to compute the new security strategy of player 1.

Our approach consists of three steps. First, considering the potentially large size of player 2's action set, in other words, a potentially large number of constraints in the LP problem, we propose to use \emph{shadow vertex simplex method}, whose average computational complexity is low with respect to the constraint size, as the basic method to solve the LP. Second, we propose a necessary and sufficient condition guaranteeing player 1's security strategy remains unchanged. If such a condition is violated, we then propose an iterative algorithm to reduce the computational time based on the transient results of the previous computation process.

\section{Shadow Vertex Simplex Method with a Probability Variable}
The shadow vertex simplex method was introduced in \cite{borgwardt1982average}, and is motivated by the observation that the simplex method is very simple in two dimensions when the feasible set forms a polygon. In this case, simplex method visits the adjacent vertices, whose total number is usually small in two dimensional cases, until it reaches the optimum vertex. The basic idea of the shadow vertex simplex method consists of two steps. It first projects a high dimensional feasible set to a two dimensional plane such that the projection of the optimum vertex is still a vertex of the projection (or shadow) of the feasible set. Then it walks along the adjacent vertices whose projections are also vertices of the shadow of the feasible set until it reaches the optimum vertex.

In order to use the shadow vertex method, we need to transform the LP (\ref{eq: old LP 1}-\ref{eq: old LP 4}) into the canonical form. The canonical form of LP (\ref{eq: new LP 1}-\ref{eq: new LP 5}) can be easily derived by replacing $G$ in (\ref{eq: old canonical LP 1}-\ref{eq: A}) with $[G\ g]$. Let $x=[\bar{x}_1\ \bar{x}_2\ \ldots,\bar{x}_{n-1}\ \ell]^T$. We have $\bar{x}=[T\ \mathbf{0}]x+e_n$, where $T=[I_{n-1}\ -\mathbf{1}]^T$, and $e_n$ is the $n$th column of $I_n$. LP (\ref{eq: old LP 1}-\ref{eq: old LP 4}) is rewritten as below.
\begin{align}
  V=&\max_{x\in \mathbb{R}^{n}}c^Tx \label{eq: old canonical LP 1}\\
  s.t.& A x \leq b, \label{eq: old canonical LP 2}
\end{align}
where $c=[\mathbf{0}^T 1]^T$, $b=[e_n^TG\ 1\ \mathbf{0}^T]^T$, and
\begin{align}
A=\left[\begin{array}{cc}
                 -G^TT & \mathbf{1} \\
                 \mathbf{1}^T & 0 \\
                 -I_{n-1} & \mathbf{0}
               \end{array}
        \right]. \label{eq: A}
\end{align}
We call the first $m$ constraints in (\ref{eq: old canonical LP 1}-\ref{eq: old canonical LP 2}) the normal constraints, and the last $n$ constraints the probability constraints. Notice that the last $n$ constraints cannot be active at the same time.

\subsection{Initial shadow vertex, projection plane, and initial searching table}
Denote the feasible set of LP (\ref{eq: old canonical LP 1}-\ref{eq: old canonical LP 2}) by $X$. We set $x_0=[\mathbf{0}\ \min_{i=1,\ldots,m}G_{n,i}]^T$ to be the initial vertex. Given a feasible solution $x\in \mathbb{R}^n$, we say constraint $i$ is active if $A_ix=b_i$. Let $\Omega_0=\{i:A_ix_0=b_i\}$ be the initial active constraint set, where $A_i$ is the $i$th row of $A$. Next, we will introduce the necessary and sufficient condition of optimality in terms of active constraint set. The theorem is similar to lemma 1.1 of \cite{borgwardt2012simplex} except that we don't require $b$ to be positive.
\begin{theorem}
  \label{theorem: optimality condition primal}
Consider a general linear program in canonical form.
\begin{align*}
  &\max_{x\in \mathbb{R}^n} c^Tx\\
  s.t.& Ax\leq b,
\end{align*}
where $A\in\mathbb{R}^{m\times n}$.
Let $x_0$ be a vertex and $\Omega_0$ be the corresponding active constraint set. Then $x_0$ is maximal with respect to $w^Tx$ for some $w\in\mathbb{R}^n\setminus \{0\}$ if and only if there exists a non-negative vector $\rho\in \mathbb{R}^{|\Omega_0|}$ such that
\begin{align}
  w=\sum_{i=1}^{|\Omega_0|} \rho_i A_{\Omega_0^i}^T, \label{eq: optimality condition}
\end{align}
where $\Omega_0^i$ is the $i$th element in $\Omega_0$.
\end{theorem}
\begin{proof}
Assume that equation (\ref{eq: optimality condition}) holds. We know that $A_{\Omega_0^i}x_{0}=b_{\Omega_0^i}\geq A_{\Omega_0^i}x$ for any feasible $x$, which implies that $w^Tx_0\geq w^T x$ for any feasible $x$.

Now assume that $w^Tx_0\geq w^T x$ for any feasible $x$. According to the inhomogeneous Farkas Theorem \cite{Iouditski2016}, we derive that $w=\sum_{i=1}^{m}\rho_iA_i^T$ for some nonnegative vector $\rho$ satisfying $\sum_{i=1}^m \rho_i b_i\leq w^T x_0=\sum_{i=1}^{m}\rho_iA_i x_0$. It implies that $\sum_{i\notin \Omega_0}\rho_i(b_i-A_ix_0) \leq 0$. Since $\rho_i\geq 0$ and $b_i-A_ix_0>0$ for $i\notin\Omega_0$, the inequality holds only if $\rho_i=0$ for $i\notin\Omega_0$, and equation (\ref{eq: optimality condition}) is true.
\end{proof}

From theorem \ref{theorem: optimality condition primal}, we see that searching for the optimum vertex is the same as looking for the active constraints whose convex cone contains the objective vector $c$. In contrast with Dantzig'g simplex method, which modifies basic variables iteratively until optimality conditions are satisfied, the shadow vertex method searches for active constraints. To make the searching process efficient, shadow vertex method limits its search to shadow vertices which are defined as below.
\begin{definition}
  \label{def: shadow vertex}
 Let $u,v$ be two linearly independent vectors and $\Gamma$ be the orthogonal projection onto $span(u,v)$. A vertex $x$ of a polygon $X$ is called a shadow vertex with respect to $span(u,v)$ if $\Gamma(x)$ is a vertex of $\Gamma(X)$.
\end{definition}
An important issue is how to design the two dimensional projection plane $span(u,v)$ such that both the initial and the optimum vertex are shadow vertices. To this end, we introduce the relaxed non-degeneracy assumption and a necessary and sufficient condition for a vertex of the feasible set to be a shadow vertex.

\begin{assumption}
\label{assumption: nondegeneracy}
Every $n$ element subset of $\{A_1,A_2,\ldots,A_{m+n},u,c\}$, where at most $n-1$ elements are from $\{A_{m+1},\ldots,A_{m+n}\}$, is linearly independent, and $\{A_1,A_2,\ldots,A_{m+n}\}$ is in general position with respect to $b$, i.e. for any $n$ element subset $\Omega$ of $\{1,\ldots,m+n\}$ such that there exists an $x\in \mathbb{R}^n$ satisfying $A_\Omega x=b_\Omega$, $A_i x\neq b_i$ for all $i\notin\Omega$.
\end{assumption}
Assumption \ref{assumption: nondegeneracy} guarantees that for any vertex $x$, there are only $n$ active constraints. The main difference between the relaxed non-degeneracy assumption and the non-degeneracy assumption in \cite{borgwardt2012simplex} is that we allow the last $n$ constraints to be linearly dependent. This difference doesn't influence the necessary and sufficient condition of a shadow vertex (Lemma 1.2 in \cite{borgwardt2012simplex}), and the proof is also similar. For the completeness of this paper, we give the theorem and the proof as follows.
\begin{theorem}
\label{theorem: shadow vertex}
Consider LP problem (\ref{eq: old canonical LP 1}-\ref{eq: old canonical LP 2}), and suppose Assumption \ref{assumption: nondegeneracy} holds. Let $x_0$ be a vertex of the feasible set $X$ of (\ref{eq: old canonical LP 2}), and $\Gamma: X \rightarrow span(u,v)$ be the orthogonal projection map from $X$ to $span(u,v)$. The following three conditions are equivalent.
\begin{enumerate}
  \item $x_0$ is a shadow vertex.
  \item The projection of $x_0$ is on the boundary of the projection of $X$, i.e. $\Gamma(x_0)\in \partial\Gamma(X)$.
  \item There exists a vector $w\in\mathbb{R}^n\setminus \{\mathbf{0}\}$ in $span(u,v)$ such that $w^Tx_0=\max_{x\in X} w^Tx$.
\end{enumerate}
\end{theorem}
\begin{proof}
1) $\Rightarrow$ 2): It is clear that if $x_0$ is a shadow vertex, its projection lies in the boundary of $\Gamma(X)$.

2) $\Rightarrow$ 3): Let $\Gamma(x_0)\in \partial \Gamma(X)$. Because of the convexity of the $\Gamma(X)$, there must exist a $w\in span(u,c)\setminus \{\mathbf{0}\}$ such that $w^T\Gamma(x_0) \geq w^T \Gamma(x)$ for any $x\in X$. Meanwhile, we know that $x-\Gamma(x)\bot span(u,c)$, and hence $w^Tx=w^T\Gamma(x)$ for any $x\in X$. Therefore, we have $w^Tx_0\geq w^Tx$ for any $x\in X$.

3) $\Rightarrow$ 1): Now let's assume that 3) is true. Because $x-\Gamma(x)\bot span(u,c)$ and $w\in span(u,c)$, we know that $\Gamma(x_0)$ is the maximal relative to $w$ for $x\in \Gamma(X)$. Therefore, $\Gamma(x_0)$ is in the boundary of the shadow $\Gamma(X)$. Since $\Gamma(X)$ is a two-dimensional polygon, if $\Gamma(x_0)$ is not a vertex, then it must lies inside an edge.

Together with the fact that $\Gamma(x_0)$ is the optimum w.r.t $w$, we know that $w$ is orthogonal to the edge, and there exists a $v\in span(u,c)\setminus \{0\}$ such that $w\bot v$ and $\Gamma(x_0)$ is the maximal relative to $w+\epsilon v$ for any $x\in \Gamma(X)$ if and only if $\epsilon=0$.
Since $x-\Gamma(x)\bot span(u,c)$ and $w+\epsilon v\in span(u,c)$, we know that $x_0$ is also maximal relative to $w+\epsilon v$ if and only if $\epsilon=0$ for any $x\in X$.

Let $\Omega_0$ be the active constraint set when $x=x_0$. Assumption \ref{assumption: nondegeneracy} indicates that there are $n$ elements in $\Omega_0$. Since $x_0$ is maximal relative to $w$, according to Theorem \ref{theorem: optimality condition primal}, $w=\sum_{i=1}^n \rho_i A_{\Omega_0^i}^T\neq 0$ for $\rho\geq 0$. Let $v=\sum_{i=1}^n \alpha_i A_{\Omega_0^i}^T$, and hence $w+\epsilon v=\sum_{i=1}^n (\rho_i+\epsilon \alpha_i)A_{\Omega_0^i}^T$. $x_0$ is not maximal relative to $w+\epsilon v$ for $\epsilon>0$ implies that there exists an $l$ such that $\rho_l+\epsilon \alpha_l<0$ for any $\epsilon>0$, and from the continuity of the function, we see that $\rho_l+0\alpha_l=0$. Similarly, $x_0$ is not maximal relative to $w+\epsilon v$ for $\epsilon<0$ implies that there exists a $k$ such that $\rho_k+\epsilon \alpha_k<0$ for any $\epsilon<0$, and the continuity of the function implies that $\rho_k+0\alpha_k=0$. Moreover, together with the fact that  $\rho_i+\epsilon \alpha_i$ is a linear function of $\epsilon$ for $i=l,k$, we see that $l\neq k$.

Therefore, we know that at least $2$ elements of $\rho$ is $0$, and we have
\begin{align}
  w=\beta_1 u+\beta_2 c=\sum_{i=1,i\neq k,l}^n \rho_i A_{\Omega_0^i}^T. \label{eq: proof 1}
\end{align}
Since constraint $m+1,\ldots,m+n$ cannot be active at the same time, $\Omega_0$ contains at most $n-1$ elements from $\{m+1,m+2,\ldots,m+n\}$. Therefore, equation (\ref{eq: proof 1}) contradicts Assumption \ref{assumption: nondegeneracy}, and $\Gamma(x_0)$ is a vertex of $\Gamma(X)$.
\end{proof}

Theorem \ref{theorem: shadow vertex} implies that if we choose the projection plane to contain the objective vector $c$, then the optimum vertex is a shadow vertex. According to Theorem \ref{theorem: optimality condition primal}, if we construct an auxiliary objective vector $u=\sum_{i=1}^{|\Omega_0|}\rho_i A_{\Omega_0^i}^T\neq 0$ for some non-negative $\rho_i$'s, which is linearly independent of $c$, then the initial vertex $x_0$ is optimal with respect to $u$, and hence $x_0$ is a shadow vertex with respect to $span(u,c)$ according to the shadow vertex condition (\ref{theorem: shadow vertex}).

Now that we have found the initial vertex, the initial active constraint set and the projection plane, it is time to build the searching table. The basic idea of constructing the table is to find a linear combination of active constraint vectors for objective vector $c$, auxiliary objective vector $u$ and all other constraint vectors. Given any active constraint set $\Omega_0$, we construct the searching table in the following way. The first row consists of an $n$-dimensional row vector $\alpha$ and a scalar $Q_c$ which satisfies $c=\sum_{j=1}^n \alpha_j A_{\Omega_0^j}$ and $Q_c=-\sum_{j=1}^n \alpha_j b_{\Omega_0^j}$. The second row consists of an $n$-dimensional row vector $\beta$ and a scalar $Q_u$ which satisfies $u=\sum_{j=1}^n \beta_j A_{\Omega_0^j}$ and $Q_u=-\sum_{j=1}^n \beta_j b_{\Omega_0^j}$. Each row for the next $m+n$ rows consists of an $n$-dimensional row vector $\gamma_i$ and a scalar $\phi_i$ which satisfies $A_i=\sum_{j=1}^n \gamma_{ij}A_{\Omega_0^j}$ and $\phi_i=b_i-\sum_{j=1}^n \gamma_{ij}b_{\Omega_0^j}$. Notice that the objective value is $-Q_c$ and feasibility is implied by non-negativity of $\phi_j$'s. We build the initial table using the following algorithm.

\begin{algorithm}[Initialization]\hfill{}
\label{algorithm: initialization from scrach}
\begin{enumerate}
  \item Find $l$ such that $G_{n,l}=\min_j G_{n,j}$.
  \item Let $\Omega_0=\{l,m+2,\ldots,m+n\}$.
  \item Find $\beta>0$ such that $u=\sum_{j=1}^n \beta_j A_{\Omega_0^j}$ satisfies Assumption \ref{assumption: nondegeneracy}.
  \item Compute $Q_u=-\sum_{j=1}^n \beta_j b_{\Omega_0^j}$.
  \item Compute $\alpha$ such that $c=\sum_{j=1}^n \alpha_j A_{\Omega_0^j}$.
    \begin{align*}
      \alpha_j=\left\{\begin{array}{cl}
                        1,& \hbox{if $j=1$} \\
                        A_{l,j-1},& \hbox{if $j\neq 1$}
                      \end{array}
      \right.
    \end{align*}
  \item Compute $Q_c=-\sum_{j=1}^n \alpha_j b_{\Omega_0^j}$.
  \item Compute $\gamma_i$ such that $A_i=\sum_{j=1}^n \gamma_{ij}A_{\Omega_0^j}, \forall i=1,\ldots,m+n$.
    \begin{align*}
      \gamma_{ij}=&\left\{\begin{array}{cl}
                        1,& \hbox{if $j=1$ and $i\leq m$} \\
                        A_{l,j-1}-A_{i,j-1},& \hbox{if $j\neq 1$ and $i\leq m$.}
                      \end{array}
      \right.\\
      \gamma_{m+1}=&[0\ -1\ \ldots\ -1],\\
      \gamma_{m+j}=&e_j, \forall j=2,\ldots,n.
    \end{align*}
  \item Compute $\phi_i=b_i-\sum_{j=1}^n \gamma_{ij}b_{\Omega_0^j}$, for all $i=1,\ldots,m+n$.
  \item Record the initial table $table(0)$ and the initial active constraint set $\Omega_0$.
\end{enumerate}
\end{algorithm}

\subsection{Pivot step}
Suppose the current shadow vertex is $x_t$, and the corresponding table and active constraint set $\Delta_t$ are given. We first check whether $x_t$ is the optimum vertex. If not, we search for the adjacent shadow vertex with a larger objective value (if exists). This is called 'taking a pivot step'.

According to Theorem \ref{theorem: optimality condition primal}, if $\alpha_i\geq 0$ for all $i=1,\ldots,n$, then $x_t$ is the optimum vertex, and the search ends. Otherwise, we need to find an adjacent shadow vertex with a larger objective value. To this end, the following lemma is useful.

\begin{lemma}[Lemma 1.4 in \cite{borgwardt2012simplex}]
\label{lemma: pivot direction}
Consider LP problem (\ref{eq: old canonical LP 1}-\ref{eq: old canonical LP 2}). Suppose $u$ is linearly independent with $c$. Let $x_0, x_1,\ldots,x_t$ be the maximal vertices with respect to $w_0^T x, w_1^Tx,\ldots,w_t^Tx$, where $w_i\in span (u,c)\setminus \{0\}$, $w_0=u$, and $arc(w_i,c)>arc(w_{i+1},c)$ for $i=0,\ldots,t-1$. Then
\begin{align*}
 c^T x_i < c^Tx_{i+1}, \forall i=0,\ldots,t-1.
\end{align*}
\end{lemma}

With Lemma \ref{lemma: pivot direction}, we build a vector $w(\mu)=u+\mu c$. Let $w_t=w(\mu_t)$ where $\mu_t$ is the smallest non-negative $\mu$ such that the current shadow vertex $x_t$ is maximal with respect to $w(\mu_t)^Tx$. If $t=0$, then $\mu_t=0$. The basic idea of pivot is to increase $\mu$ from $\mu_t$ to $\mu_{t+1}$ such that $x_t$ won't be the maximum with respect to $w(\mu_{t+1}+\epsilon)$ for any $\epsilon>0$, and then find the next shadow vertex $x_{t+1}$ which is maximal with respect to $w_{t+1}=w(\mu_{t+1})$. Since $arc(w_t,c)>arc(w_{t+1},c)$, we assure that $x_{t+1}$ has a larger objective value than $x_t$ according to Lemma \ref{lemma: pivot direction}.

Now, let us take a closer look on how the pivot step is done. First, we rewrite $w(\mu)=\sum_{i=1}^n (\beta_i+\mu \alpha_i)A_{\Omega_t^i}$. Since $x_t$ is not maximal with respect to $c$, there must exist an $i\in \{1,\ldots,n\}$ such that $\alpha_i<0$, and $\beta_i+\mu\alpha_i$ decreases to negative as $\mu$ increases. If $\beta_k+\mu\alpha_k$ first decreases to $0$ for some $\mu=\mu_{t+1}$, then we shall replace $\Omega_t^k$ with another constraint out of $\Omega_t$ such that the new vertex $x_{t+1}$ formed by the new active constraint set $\Omega_{t+1}$ is the maximal relative to $w(\mu_{t+1}+\epsilon)=u+(\mu_{t+1}+\epsilon) c $ where $\epsilon$ is an arbitrary small positive number. In this way, the moving out active constraint $\Omega_t^k$ is figured out where
\begin{align}
  k=\argmin_{i\in\{1,\ldots,n\},\alpha_i<0}-\frac{\beta_i}{\alpha_i}. \label{eq: moving out index}
\end{align}

The above analysis elucidates which constraint in the active constraint set should be replaced to increase the objective value. Next, we will discuss which inactive constraint should be moved into the active constraint set to guarantee feasibility. First of all, the moving-in constraint $i$ should satisfy $\gamma_{ik}<0$. If $\gamma_{ik}\geq 0$ for all $i$, then $A_k$ and all $A_j$'s for $j\notin \Omega$ lies in the same half space divided by the hyperplane through points $\mathbf{0}, A_{\Omega_t^1},\ldots,A_{\Omega_t^{k-1}},A_{\Omega_t^{k+1}},\ldots,A_{\Omega_t^n}$, and $c$ lies on the other half space. Therefore, $c$ doesn't lie in the convex cone of $A_1,\ldots,A_{m+n}$, which means that the LP problem has no solution, and the search stops. If there exists at least one inactive constraint $i$ such that $\gamma_{ik}<0$, let us suppose that we choose constraint $l$ satisfying $\gamma_{lk}<0$ as the moving in constraint. The table will be updated as follows \cite{borgwardt2012simplex}.
\begin{align}
  \alpha_j=&\left\{\begin{array}{cl}
    \alpha_j-\alpha_k\frac{\gamma_{lj}}{\gamma_{lk}} & \hbox{if $j\neq k$}\\
    \frac{\alpha_k}{\gamma_{lk}} & \hbox{if $j=k$}
  \end{array}\right. \label{eq: tabular update 1}\\
  \beta_j=&\left\{\begin{array}{cl}
    \beta_j-\beta_k\frac{\gamma_{lj}}{\gamma_{lk}} & \hbox{if $j\neq k$}\\
    \frac{\beta_k}{\gamma_{lk}} & \hbox{if $j=k$}
    \end{array}
  \right.\label{eq: tabular update 2}\\
  \gamma_{ij}=&\left\{\begin{array}{cl}
    \gamma_{ij}-\gamma_{ik}\frac{\gamma_{lj}}{\gamma_{lk}} & \hbox{if $j\neq k$ and $i\notin \Omega_{t+1}$}\\
    \frac{\gamma_{ik}}{\gamma_{lk}} & \hbox{if $j=k$ and $i\notin \Omega_{t+1}$}
    \end{array}
  \right.\label{eq: tabular update 3}\\
  \gamma_i=&e_k,\hbox{if $i=l$}. \label{eq: tabular update 4}
\end{align}
and
\begin{align}
  \phi_i=&
  \left\{\begin{array}{cl}
    0 & \hbox{if $i\in \Omega_{t+1}$}\\
    \phi_i-\frac{\gamma{ik}}{\gamma_{lk}}\phi_l & \hbox{if $i\notin\Omega_{t+1}$}
  \end{array}
  \right.\label{eq: tabular update 5}\\
  Q_c=& Q_c-\frac{\phi_l \alpha_k}{\gamma_{lk}} \label{eq: tabular update 6}\\
  Q_u=& Q_u-\frac{\phi_l \beta_k}{\gamma_{lk}} \label{eq: tabular update 7}
\end{align}
As mentioned before, the non-negativity of $\phi_i$'s implies feasibility. Therefore, according to (\ref{eq: tabular update 5}) and the analysis above, the moving in constraint $l$ is chosen such that
\begin{align}
  l=\argmax_{i\notin\Omega_t,\gamma_{ik}<0}\frac{\phi_{i}}{\gamma_{ik}}, \label{eq: moving in constraint}
\end{align}
and the active constraint set is updated to $\Omega_{t+1}=\{\Omega_t^1,\ldots,\Omega_t^{k-1},l,\Omega_t^{k+1},\ldots,\Omega_t^n\}$.

The pivot algorithm given the current active constraint set $\Omega_t$ is provided as follows.
\begin{algorithm}[Pivot]\hfill{}
  \label{algorithm: pivot}
 \begin{enumerate}
   \item If $\alpha_i\geq 0$ for all $i$, then the vertex associated with $\Omega_t$ is the optimum, and $v=-Q_c$. Go to step 8).
   \item Find the moving-out constraint $\Omega_t^k$, where $k$ is given in (\ref{eq: moving out index}).
   \item If $\gamma_{ik}\geq 0$ for all $i\notin \Omega_t$, then there is no solution. Go to step 8).
   \item Find the moving-in constraint $l$ satisfying (\ref{eq: moving in constraint}).
   \item Update and record $\Omega_{t+1}=\{\Omega_t^1,\ldots,\Omega_t^{k-1},l,\Omega_t^{k+1},\ldots,\Omega_t^n\}$.
   \item Update the table according to (\ref{eq: tabular update 1}-\ref{eq: tabular update 7}), and record it as $table(t+1)$.
   \item Return to step 1).
   \item End.
 \end{enumerate}
\end{algorithm}

Shadow vertex simplex method has a polynomial computational complexity \cite{borgwardt2012simplex}. To be more specific, let $\tau$ be the number of pivot steps, and $T$ be the number of shadow vertices. If $A_1,A_2,\ldots,A_m$ and $c$ are independently, identically, and symmetrically distributed under rotation, then $E(t)\leq E(T)\leq \rho m^{\frac{1}{n-1}} n^3$ for some positive constant $\rho$. 
%


\section{Player 1's Unchanging security strategy}
\label{sec: unchanging strategy}
It is not necessary to run the shadow vertex simplex method every time player 2 adds a new action since player 2's new action may have no influence on player 1's security strategy and the game value. Comparing the old LP (\ref{eq: old LP 1}-\ref{eq: old LP 4}) with the new LP (\ref{eq: new LP 1}-\ref{eq: new LP 5}), we see that the new LP adds a new constraint (\ref{eq: new LP 3}) to the old LP. Geometrically, a new constraint means a new cut of the existing feasibility set. If the new constraint does not cut the optimum vertex off, i.e. the optimum vertex satisfies the new constraint, then the optimum vertex remains the same, and player 1's security strategy does not change. To formally state the above analysis, we first provide the canonical form of the new LP (\ref{eq: new LP 1}-\ref{eq: new LP 5}).

As mentioned before, the canonical form of (\ref{eq: new LP 1}-\ref{eq: new LP 5}) takes the same form as in (\ref{eq: old canonical LP 1}-\ref{eq: old canonical LP 2}) except that the payoff matrix is replaced with $[G\ g]$. To make the discussion clear, we provide the canonical form of (\ref{eq: new LP 1}-\ref{eq: new LP 5}) as below.
\begin{align}
  V_{new}=&\max_{x\in\mathbb{R}^n} c^T x \label{eq: new canonical form 1}\\
  s.t.& \hat{A}x \leq \hat{b},\label{eq: new canonical form 2}
\end{align}
where $$\hat{A}=[A^T_{1:m}\ [g_n-g_1\ g_n-g_2\ \ldots\ g_n-g_{n-1}\ 1]^T A^T_{m+1:m+n}]^T,$$ and $$\hat{b}=[b_{1:m}^T\ g_n\ b_{m+1:m+n}^T]^T.$$ The new constraint in the new LP (\ref{eq: new LP 1}-\ref{eq: new LP 5}) is transformed to
\begin{align}
[g_n-g_1\ g_n-g_2\ \ldots\ g_n-g_{n-1}\ 1]x \leq g_n .\label{eq: new constraint}
\end{align}

\begin{theorem}
\label{theorem: optimality reserved}
Suppose Assumption \ref{assumption: nondegeneracy} holds. Let $x^*$ be the optimal solution of old LP (\ref{eq: old canonical LP 1}-\ref{eq: old canonical LP 2}). The old optimal solution $x^*$ remains optimal in the new LP (\ref{eq: new canonical form 1}-\ref{eq: new canonical form 2}) if and only if $x^*$ satisfies the new constraint (\ref{eq: new constraint}).
\end{theorem}
\begin{proof}
It is easy to see that if $x^*$ is the optimal solution of (\ref{eq: new canonical form 1}-\ref{eq: new canonical form 2}), then $x^*$ satisfies the new constraint. For the opposite direction, with one more constraint, the feasible set of new LP (\ref{eq: new canonical form 1}-\ref{eq: new canonical form 2}) is included in the feasible set of old LP (\ref{eq: old canonical LP 1}-\ref{eq: old canonical LP 2}). So we have $V_{new}\leq V$, in other words, $x^+_n \leq x^*_n$, where $x^+$ indicates the optimal solution of (\ref{eq: new canonical form 1}-\ref{eq: new canonical form 2}). Meanwhile, if $[g_n-g_1\ g_n-g_2\ \ldots\ g_n-g_{n-1}\ 1]x^* \leq g_n$, it is easy to see that $x^*$ is a feasible solution of (\ref{eq: new canonical form 1}-\ref{eq: new canonical form 2}), and we have $x^*_n \leq x^+_n$. Therefore, $x^*_n=x^+_n$, and $x^*$ is the optimal solution of (\ref{eq: new canonical form 1}-\ref{eq: new canonical form 2}).
\end{proof}

Since the re-computation is not always necessary, we are interested in the probability of the re-computation. 
\begin{assumption}
  \label{assp: unique optimum}
The old LP problem (\ref{eq: old canonical LP 1}-\ref{eq: old canonical LP 2}) has a unique optimal solution.
\end{assumption}

\begin{assumption}
  \label{assp: equal probability active set}
Consider the canonical form of the new LP problem (\ref{eq: new canonical form 1}-\ref{eq: new canonical form 2}). Suppose Assumption \ref{assumption: nondegeneracy} holds. Any $n$-element subset $\Delta$ of $\{1,\ldots,m,m+1,m+2\ldots,m+1+n\}$ such that $A_{\Delta^1},\ldots,A_{\Delta^n}$ are linearly independent has the same probability to be the active constraint set with respect to the optimum of the new LP problem (\ref{eq: new canonical form 1}-\ref{eq: new canonical form 2}), i.e.
\begin{align}
  P(\hbox{$\Delta$ is an active constraint set})=\frac{1}{\binom{m+1+n}{n}-1}, \label{eq: equal possibility}
\end{align}
where $|\Delta|=n$ and $A_{\Delta^1},\ldots,A_{\Delta^n}$ are linearly independent
\end{assumption}
The relaxed non-degeneracy assumption \ref{assumption: nondegeneracy} indicates that there are $n$ elements in every active constraint set, and any $n$ rows, except the last $n$ rows, of $A$ are linearly independent. Therefore, there are $\binom{m+1+n}{n}-1$ possible active constraint sets, and since each set has the same probability to be active, we have equation (\ref{eq: equal possibility}).

\begin{lemma}
  \label{lem: active new constraint}
Suppose Assumption \ref{assumption: nondegeneracy}, \ref{assp: unique optimum} and \ref{assp: equal probability active set} hold. The unique optimal solution $x^*$ of the old LP (\ref{eq: old canonical LP 1}-\ref{eq: old canonical LP 2}) is not optimal any more in the new LP (\ref{eq: new canonical form 1}-\ref{eq: new canonical form 2}) if and only if the new constraint (\ref{eq: new constraint}) is an active constraint for any optimum of the new LP (\ref{eq: new canonical form 1}-\ref{eq: new canonical form 2}), i.e. $[g_n-g_1\ g_n-g_2\ \ldots\ g_n-g_{n-1}\ 1]x^+ = g_n $, where $x^+$ is any optimum of the new LP (\ref{eq: new canonical form 1}-\ref{eq: new canonical form 2}).
\end{lemma}
\begin{proof}
Suppose $x^*$ is not an optimum of the new LP. According to Theorem \ref{theorem: optimality reserved}, we have $\hat{A}_{m+1}x^*>\hat{b}_{m+1}$. If there exists an optimum $x^+ \neq x^*$ of the new LP such that $\hat{A}_{m+1}x^+<\hat{b}_{m+1}$, then there must exist an $x=\alpha x^+ +(1-\alpha) x^*$ such that $\hat{A}_{m+1}x=\hat{b}_{m+1}$. It is easy to verify that $x$ is a feasible solution of the new LP. Since $x^+$ is an optimum of the new LP, we have $c^Tx^+ \geq c^T x$. Meanwhile, $x^*$ is the unique optimum of old LP implies that $c^Tx^+<c^Tx^*$, and we have $c^Tx=\alpha c^Tx^+ +(1-\alpha) c^Tx^*>c^T x^+$, which is a contraction. Therefore, the new constraint is active for any optimum of the new LP.

Suppose the new constraint is active for any optimum of new LP (\ref{eq: new canonical form 1}-\ref{eq: new canonical form 2}). The old optimum $x^*$ cannot be an optimum of the new LP, since its corresponding active constraint set only contains the old constraints.
\end{proof}

\begin{theorem}
  \label{thm: unchanging probability}
Consider the canonical form of the new LP problem (\ref{eq: new canonical form 1}-\ref{eq: new canonical form 2}). Suppose Assumption \ref{assumption: nondegeneracy}, \ref{assp: unique optimum} and \ref{assp: equal probability active set} hold. The probability that player 1's security strategy changes is $\frac{n}{m+1+n-(m+1)/\binom{m+n}{n}}$.
\end{theorem}
\begin{proof}
Player 1's security strategy changes if and only if the optimum $x^*$ of the old LP (\ref{eq: old canonical LP 1}-\ref{eq: old canonical LP 2}) is not optimal any more in the new LP. According to Lemma \ref{lem: active new constraint}, it is the same to say that the new constraint $\hat{A}_{m+1}x\leq \hat{b}_{m+1}$ is an active constraint in the canonical form of the new LP problem. We have $P(\hbox{player 1's security strategy changes})$ $=P(\hbox{constraint $m+1$ is an active constraint})$.
\begin{align*}
  &P(\hbox{constraint $m+1$ is an active constraint})\\
  =& \sum_{\{i_1,\ldots,i_{n-1}\}\in\{1,\ldots,m,m+2,\ldots,m+1+n\}} P(\{i_1,\ldots,i_{n-1},m+1\}\hbox{ is the active constraint set})\\
  =&\binom{m+n}{n-1}\frac{1}{\binom{m+1+n}{n}-1}\\
  =&\frac{n}{m+1+n-(m+1)/\binom{m+n}{n}}.
\end{align*}
\end{proof}
Theorem \ref{thm: unchanging probability} implies that as $m$ grows, the probability to recompute the security strategy decreases. This is because as the size of constraints grows, the feasible set shrinks, and it is less possible for the new constraint to cut the optimum vertex off.

\section{Iterative shadow vertex simplex method}
\label{sec: iterative s.v. method}
If the old security strategy is not optimal any more in the new game, it is not always necessary to start the optimum search from the initial vertex. Suppose that we have constructed a complete shadow vertex path $\Pi=\{x_0,x_1,\ldots,x_{\tau-1},x^*\}$ when searching for the optimum of the old LP (\ref{eq: old canonical LP 1}-\ref{eq: old canonical LP 2}). When the new constraint (\ref{eq: new constraint}) is added, we can test the feasibility of the visited shadow vertices one by one starting from $x_0$ until we find one that fails the feasibility test. That last vertex should have the largest objective value, and will be chosen as the starting point of the new search. This idea is formulated as the iterative shadow vertex simplex method as follows.

\begin{algorithm}[Iterative shadow vertex algorithm]\hfill{}
  \label{algorithm: iterative shadow vertex method}
  \begin{enumerate}
    \item If $x^*$ satisfies (\ref{eq: new constraint}), then $x^*$ is the optimal solution. Go to step 6).
    \item Otherwise, set the new path $\Pi_{new}=\emptyset$.
    \item From $i=0$ to $\tau-1$, test whether $x_i$ satisfies the new constraint (\ref{eq: new constraint}). If yes, add $x_i$ into the new path $\Pi_{new}$, and update $table(i)$ by inserting a new row vector $\gamma_{m+1}$ and a new scalar $\phi_{m+1}$ with respect to the new constraint. Otherwise, stop the test.
    \item If $\Pi_{new}=\emptyset$, then run standard shadow vertex algorithm, i.e. run the initialization algorithm \ref{algorithm: initialization from scrach}, and then the pivot algorithm \ref{algorithm: pivot}. Add all visited shadow vertices including the optimal vertex into $\Pi_{new}$. Update $\Pi=\Pi_{new}$. Go to step 6).
    \item Otherwise, choose the last element in $\Pi_{new}$ as the initial shadow vertex, and run the pivot algorithm \ref{algorithm: pivot}. Add all visited shadow vertices including the optimal vertex to $\Pi_{new}$. Update $\Pi=\Pi_{new}$. Go to step 6).
    \item End.
  \end{enumerate}
\end{algorithm}

Since the iterative shadow vertex simplex method usually starts the search from a shadow vertex with a larger objective value, we expect that the iterative shadow vertex simplex method has less pivot steps and lower computational complexity than the shadow vertex simplex method. To analyze the computational complexity of the former, we start from the following lemma stating that if the new constraint does not cut the whole path off, then with the iterative shadow vertex method, the new constraint stays in the active constraint set until the optimum vertex is found.

\begin{lemma}
\label{lem: m+1 is always active}
Suppose Assumption \ref{assumption: nondegeneracy} holds, and the optimum $x^*$ of the old LP (\ref{eq: old canonical LP 1}-\ref{eq: old canonical LP 2}) violates the new constraint (\ref{eq: new constraint}). Let $\Pi=\{x_0,x_1,\ldots,x_{\tau-1},x^*\}$ be a complete search path of the old LP (\ref{eq: old canonical LP 1}-\ref{eq: old canonical LP 2}), i.e. $\{x_0,x_1,\ldots,x_{\tau-1},x^*\}$ is a sequence of adjacent shadow vertices, and $\{\Omega_0,\ldots,\Omega_{\tau}\}$ be the corresponding active constraint set. Denote $x_t\in \Pi$ the last shadow vertex in path $\Pi$ satisfying the new constraint, i.e
\begin{align}
  c^Tx_t=&\max_{x_i\in \Pi} c^T x_i \label{eq: 1}\\
  s.t. & \hat{A}_{m+1} x_i \leq \hat{b}_{m+1}. \label{eq: 2}
\end{align}
With the iterative shadow vertex simplex algorithm \ref{algorithm: iterative shadow vertex method}, let $\Pi_{new}=\{x_0,x_1,\ldots,x_t,\hat{x}_{t+1},$ $\ldots, \hat{x}_{\hat{\tau}-1},x^+\}$ be the search path of the new LP (\ref{eq: new canonical form 1}-\ref{eq: new canonical form 2}), and $\Omega_0,\Omega_1,\ldots, \Omega_t,\hat{\Omega}_{t+1},\ldots,$ $ \hat{\Omega}_{\hat{\tau}-1},\hat{\Omega}_{\hat{\tau}}$ be the corresponding active constraint sets. The new constraint (\ref{eq: new constraint}) stays in the active constraint set $\{\hat{\Omega}_{t+1},\ldots ,\hat{\Omega}_{\hat{\tau}}\}$, i.e. $m+1\in \hat{\Omega}_i$ for all $i=t+1,\ldots,\hat{\tau}$.
\end{lemma}
\begin{proof}
First, we show that $\hat{\Omega}_{t+1}$ contains $m+1$. Suppose in the search path $\Pi$ from step $x_t$ to $x_{t+1}$, the move-out constraint is $k$ and the move-in constraint is $l$. In the search process of the new LP from $x_t$ to $\hat{x}_{t+1}$, if the new constraint $\hat{A}_{m+1}x\leq \hat{b}_{m+1}$ is not chosen either because $\gamma_{m+1,k}\geq 0$ or because $\phi_{m+1}/\gamma_{m+1,k}$ is not the maximal one, $l$ will be chosen again to enter the active constraint set which implies that $\hat{x}_{t+1}=x_{t+1}$ based on Assumption \ref{assumption: nondegeneracy}. Since the pivot step guarantees feasibility, we have $x_{t+1}$ satisfies the new constraint. Meanwhile, we also have $c^Tx_{t+1}>c^Tx_t$, which contradicts equation (\ref{eq: 1}-\ref{eq: 2}). Therefore, $m+1$ must be in $\hat{\Omega}_{t+1}$.

Next, we show that once constraint $m+1$ enters the active constraint set, it will stay there until the optimum vertex is found. Let us suppose that $\hat{\Omega}_{t+i}$ contains $m+1$, but $\hat{\Omega}_{t+i+1}$ does not, for some $i=1,\ldots,\hat{\tau}-t-1$. Since $x_t$ and $\hat{x}_{t+i+1}$ are both shadow vertices, there must exist $w_t,w_{t+i+1}\in span(u,v)$ such that $x_t,\hat{x}_{t+i+1}$ are the maximum with respect to $w_t$ and $w_{t+i+1}$, respectively, according to Theorem \ref{theorem: shadow vertex}. The pivot rule of shadow vertex simplex method guarantees that $arc(w_t,c)>arc(w_{t+i+1},c)$ and $w_{t+i+1}$ is in the convex cone formed by $w_t$ and $c$. Meanwhile, $\hat{x}_{t+i+1}$ guarantees value improvement and feasibility, i.e. $c^T\hat{x}_{t+i+1}>c^T x_t$ and $\hat{A}_{m+1}\hat{x}_{t+i+1} \leq \hat{b}_{m+1}$. According to (\ref{eq: 1}-\ref{eq: 2}), we see that $\hat{x}_{t+i+1}$ is not in the search path $\Pi$ of the old LP, and $w_{t+i+1}$ does not lie in the convex cone formed by $w_t$ and $c$, which contradicts the previous conclusion. Therefore, if $\hat{\Omega}_{t+i}$ contains $m+1$, so does $\hat{\Omega}_{t+i+1}$ for all $i=1,\ldots,\hat{\tau}-t-1$.
\end{proof}

Lemma \ref{lem: m+1 is always active} implies that the iterative shadow vertex simplex method only visits shadow vertices whose active constraint set contains $m+1$, and the number of pivot steps $\hat{\tau}$ can thus be no greater than the number of shadow vertices $\hat{T}$ whose active constraint set contains $m+1$. We are interested in the expected value of $\hat{T}$. 
\begin{assumption}
  \label{assp: shadow vertex same possibility}
  Consider the canonical form of the new LP problem (\ref{eq: new canonical form 1}-\ref{eq: new canonical form 2}). Suppose Assumption \ref{assumption: nondegeneracy} holds. Any $n$-element subset $\Omega$ of $\{1,\ldots,m,m+1,m+2\ldots,m+1+n\}$ such that $A_{\Omega^1},\ldots,A_{\Omega^n}$ are linearly independent has the same probability $\xi\in (0,1)$ to form a shadow vertex with respect to $span(u,c)$.
\end{assumption}
\begin{theorem}
  \label{thm: expected hat T}
Suppose Assumption \ref{assumption: nondegeneracy} and \ref{assp: shadow vertex same possibility} hold. We have $$E(\hat{T})=\frac{n}{m+1+n-(m+1)/\binom{m+n}{n}} E(T). $$
\end{theorem}
\begin{proof}
First, notice that $$E(T)=\sum_{\{i_1,i_2,\ldots,i_n\}\in\{1,\ldots,m+1+n\}}P(\hbox{constraint $i_1,i_2,\ldots,i_n$ forms a shadow vertex}).$$ Since the last $n$ constraints can not form a vertex, together with Assumption \ref{assp: shadow vertex same possibility}, it can be shown that $E(T)=(\binom{m+1+n}{n}-1)\xi$.

Next, we have $\displaystyle E(\hat{T})=\sum_{\{i_1,\ldots,i_{n-1}\}\in\{1,\ldots,m,m+2,\ldots,m+1+n\}}P(\hbox{constraints $i_1,i_2,\ldots,i_{n-1},$}$ $\hbox{$m+1$ form a shadow vertex})=\binom{m+n}{n-1}\xi.$ Therefore, we have $$E(\hat{T})=\frac{n}{m+1+n-(m+1)/\binom{m+n}{n}} E(T). $$
\end{proof}
Notice that if the computational complexity of shadow vertex simplex method is $O(m^{\frac{1}{n-1}}n^3)$ as shown in \cite{borgwardt2012simplex}, then the computational complexity of the iterative shadow vertex method is $O(m^{\frac{1}{n-1}-1}n^4)$. If $n> 2$, then the computational complexity of the iterative shadow vertex decreases as $m$ grows. This is because as $m$ grows, it is less possible to re-compute the security strategy (Theorem \ref{thm: unchanging probability}), which results in zero pivot step and the decreased number of average pivot steps.

Under the assumption that $x^*$ violates new constraint (\ref{eq: new constraint}), we will analyze the computational complexity of both shadow vertex method and iterative shadow vertex method.
\begin{theorem}
  \label{thm: conditional computational complexity comparison}
Suppose Assumption \ref{assumption: nondegeneracy} and \ref{assp: shadow vertex same possibility} holds. Let $x^*$ be the optimal vertex of the old LP (\ref{eq: old canonical LP 1}-\ref{eq: old canonical LP 2}). If the new payoff column $g$ is independently distributed with the existing payoff matrix $G$, then $$E(\hat{T}|\hbox{$x^*$ violates (\ref{eq: new constraint})})< E(T|\hbox{$x^*$ violates (\ref{eq: new constraint})}). $$
\end{theorem}
\begin{proof}
First, we have
\begin{align*}
&E(\hat{T}|\hbox{$x^*$ does not violate (\ref{eq: new constraint})})\\
=&\sum_{\{i_1,\ldots,i_{n-1}\}\in\{1,\ldots,m,m+2,\ldots,m+1+n\}} P(\hbox{constraint $i_1,\ldots,i_{n-1},m+1$ form a shadow vertex}\\
&|\hbox{$x^*$ does not violate (\ref{eq: new constraint})}).
\end{align*}

Similarly, we have
\begin{align*}
&E(T|\hbox{$x^*$ does not violate (\ref{eq: new constraint})})\\
=&\sum_{\{i_1,\ldots,i_n\}\in\{1,\ldots,m,m+2,\ldots,m+1+n\}} P(\hbox{constraint $i_1,\ldots,i_{n}$ form a shadow vertex}\\
& |\hbox{$x^*$ does not violate (\ref{eq: new constraint})})\\
&+\sum_{\{i_1,\ldots,i_{n-1}\}\in\{1,\ldots,m,m+2,\ldots,m+1+n\}} P(\hbox{constraint $i_1,\ldots,i_{n-1},m+1$ form a shadow vertex}\\
&|\hbox{$x^*$ does not violate (\ref{eq: new constraint})})\\
=&\sum_{\{i_1,\ldots,i_n\}\in\{1,\ldots,m,m+2,\ldots,m+1+n\}} P(\hbox{constraint $i_1,\ldots,i_{n}$ form a shadow vertex}\\
& |\hbox{$x^*$ does not violate (\ref{eq: new constraint})})+E(\hat{T}|\hbox{$x^*$ does not violate (\ref{eq: new constraint})})
\end{align*}
According to Theorem \ref{theorem: shadow vertex} and \ref{theorem: optimality condition primal}, constraint $i_1,\ldots,i_{n}$ form a shadow vertex if and only if the convex cone generated by $A_{i_1},\ldots,A_{i_n}$ intersects with the plane spanned by $u$ and $c$. Since constraint $i_1,\ldots,i_{n}$ are chosen from the existing constraints which are independent with the new constraint (\ref{eq: new constraint}), the event that the convex cone generated by existing constraint vectors $A_{i_1},\ldots,A_{i_n}$ intersects with the plane spanned by $u$ and $c$ is independent with the event that the new constraint (\ref{eq: new constraint}) cuts $x^*$ off, and hence we have
\begin{align*}
  &\sum_{\{i_1,\ldots,i_n\}\in\{1,\ldots,m,m+2,\ldots,m+1+n\}} P(\hbox{constraint $i_1,\ldots,i_{n}$ form a shadow vertex}\\
  &|\hbox{$x^*$ does not violate (\ref{eq: new constraint})})\\
  =&\sum_{\{i_1,\ldots,i_n\}\in\{1,\ldots,m,m+2,\ldots,m+1+n\}} P(\hbox{constraint $i_1,\ldots,i_{n}$ form a shadow vertex})\\
  =&\sum_{\{i_1,\ldots,i_n\}\in\{1,\ldots,m,m+2,\ldots,m+1+n\}} \epsilon >0
\end{align*}
which completes the proof.
\end{proof}

The above analysis assumes that we have a complete search path $\Pi$ of the old LP. Sometimes, the new constraint may cut off a part of the search path but not the optimal vertex. In this case, we can immediately provide the optimal vertex, and then repair the search path in the background to get prepared for player 2's next new action. The idea of repairing the search path is similar to the iterative shadow vertex method, and we provide the algorithm as follows.
\begin{algorithm}[Search path repair algorithm] \hfill{}
  \begin{enumerate}
    \item If $x^*$ violates the new constraint (\ref{eq: new constraint}), go to step 8).
    \item Set $\Pi_{new}=\emptyset$ and $table_{new}=\emptyset$.
    \item From $i=0$ to $\tau-1$, if $x_i$ satisfies the new constraint, then add $x_i$ into the new path $\Pi_{new}$ and update $table_{new}(i)$ by inserting a new row vector $\gamma_{m+1}$ and a new scalar $\phi_{m+1}$ with respect to the new constraint to $table(i)$. Otherwise, stop the test.
    \item If all the elements in $\Pi$ satisfies the new constraint, go to step 8).
    \item Repair the search path.
        \begin{enumerate}
            \item Let $j=-1$.
            \item Let $j=j+1$. Choose the last element in $\Pi_{new}$ as the initial shadow vertex, and run step 2)-7) of the pivot algorithm, Algorithm \ref{algorithm: pivot}. Add the visited shadow vertex $x_{i+j}$ to $\Pi_{new}$ and update $table_{new}(i+j)$.
            \item If $x_{i+j} \in \Pi$ or $x_{i+j}$ is the optimal vertex, then go to step 6). Otherwise, go to step b).
        \end{enumerate}
    \item Add all the visited shadow vertices after $x_{i+j}$ in path $\Pi$ to the new path $\Pi_{new}$. Update the corresponding tables by inserting a new row vector $\gamma_{m+1}$ and a new scalar $\phi_{m+1}$ with respect to the new constraint, and add the updated tables to $table_{new}$.
    \item Let $table=table_{new}$, and $\Pi=\Pi_{new}$.
    \item End.
  \end{enumerate}
\end{algorithm}

\section{Numerical examples}
In this section, we consider several numerical examples to demonstrate the computational properties of the iterative shadow vertex method, and compare them to our theoretical predictions in Section \ref{sec: unchanging strategy} and \ref{sec: iterative s.v. method}.

We first generate a $10\times 100$ random payoff matrix whose elements are identically, independently, and uniformly distributed among the integers from $-100$ to $100$, then solve the corresponding zero-sum game using regular shadow vertex method, and record all visited shadow vertices and the corresponding tables. Then, the column player generates a new action, and hence produces a new random payoff column whose elements are identically, independently, and uniformly distributed among the integers from $-100$ to $100$. Notice that the newly generated payoff column is independent of the existing payoff matrix. After the new action is generated, we first decide whether re-computation of the security strategy is necessary according to Theorem \ref{theorem: optimality reserved}. If so, we use both regular shadow vertex method and iterative shadow vertex method to find the new security strategy, and record the numbers of pivot steps of both methods. This experiment is run $500$ times. Following the same steps, we also test iterative shadow vertex method in $10\times 200$, $10\times 300$, $\ldots$, $10\times 1000$ payoff matrices. The probability of re-computing the security strategy, the average number of pivot steps, and the average number of pivot steps conditioned on re-computation are given in the plots in Figure \ref{fig: main results}, where $x$-axis is the size of the action set of player 2.
\begin{figure}
  \centering
  \includegraphics[width=\textwidth]{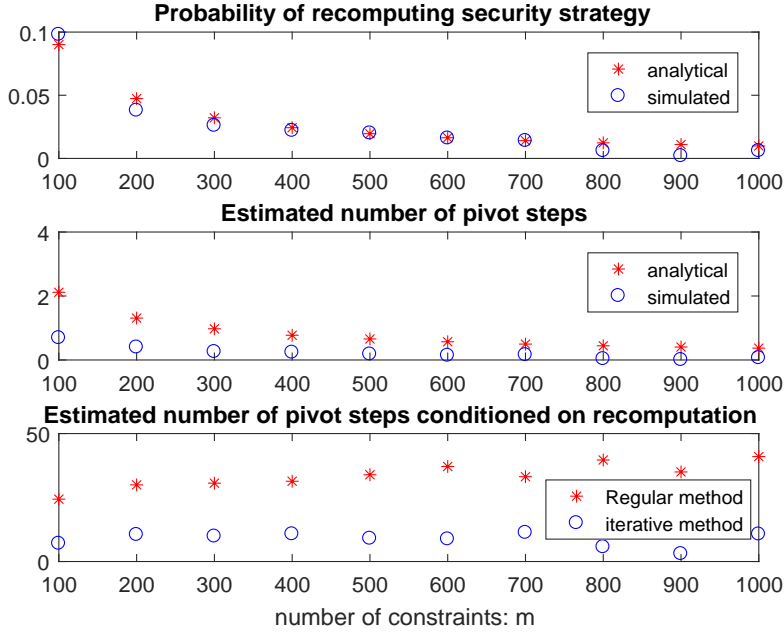}\\
  \caption{The x-axis is the original size of player 2's action set, and the size of player 1's action set is fixed to be 10. The top plot is the probability that player 1's security strategy changes when player 2 adds a new action, where red stars are the theoretical results according to Theorem \ref{thm: unchanging probability}, and blue circles are probability based on the simulation results. The middle plot is the expected number of pivot steps of iterative shadow vertex method, where red stars are the theoretical upper bounds computed according to Theorem \ref{thm: expected hat T}, and blue circles are the average number of pivot steps based on simulation results. The bottom plot is the average number of pivot steps conditioned on recomputation of player 1's security strategy based on simulation results, where red stars are the average number of pivot steps using regular shadow vertex method, and blue circles are the average number of pivot steps using iterative shadow vertex method.}
  \label{fig: main results}
\end{figure}

The top plot of Figure \ref{fig: main results} shows the probability that player 1's security strategy changes. The $x$-axis is the original size of the action set of player 2. The red stars are the analytical probability computed according to Theorem \ref{thm: unchanging probability}, and the blue circles are the empirical probability derived from the simulation results. We see that the simulated results match the analytical results. Meanwhile, we also notice that the probability of re-computing security strategy is decreasing with respect to $m$, the size of the action set of player 2, which meets our expectation.

The middle plot of Figure \ref{fig: main results} gives the average number (blue circles) of pivot steps of iterative shadow vertex method, and the appropriately scaled average number of (red stars) of pivot steps of regular shadow vertex method in accordance with the result in Theorem \ref{thm: expected hat T}. We see that the blue circles and red stars decreases in the same manner as the size of player 2's action set grows. It matches the result shown in Theorem \ref{thm: expected hat T} which indicates that the computational complexity of iterative shadow vertex method and the appropriately scaled computational complexity of regular shadow vertex method are the same.

The bottom plot of Figure \ref{fig: main results} shows that the average number (blue circles) of pivot steps of iterative shadow vertex method is always less than the average number (red stars) of pivot steps of regular shadow vertex method conditioned on the situation that player 1 needs to re-compute the security strategy, which agrees with the theoretical predictions of theorem \ref{thm: conditional computational complexity comparison}.

\section{Case study: urban security problem}
We now consider a more applied example inspired by the urban security scenario introduced in \cite{tsai2010urban}. In this paper, an urban area is modeled as a graph where edges denote roads, and vertices denote places of interest. This area have several main entrances denoted as source nodes in the graph, and several targets that attackers want to attack. With limited resources, defenders need to set up checkpoints on edges. If a checkpoint is in the path of the attackers, defenders get a corresponding reward. Otherwise, penalty is issued to the defenders. 

\begin{figure}
  \centering
  \includegraphics[width=\textwidth]{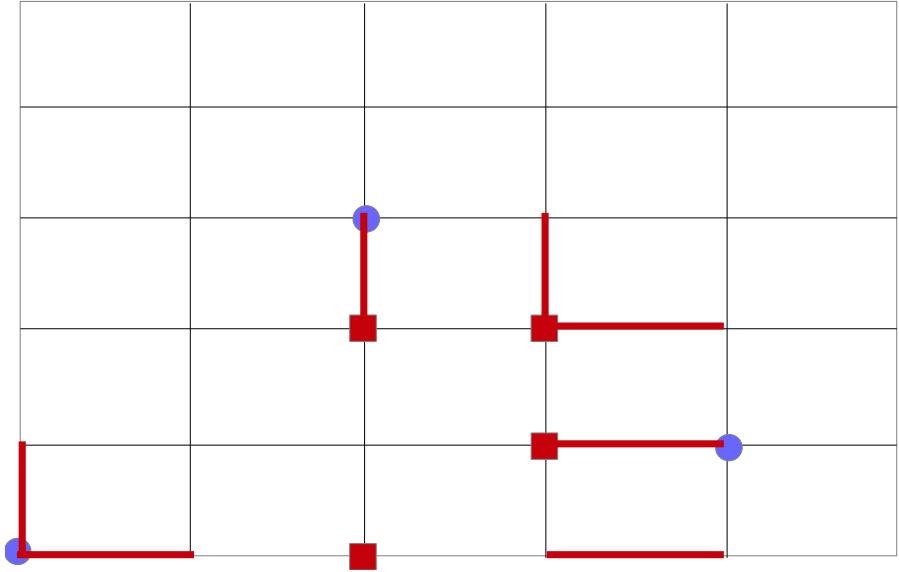}\\
  \caption{The map of an urban area, where edges and vertices denote roads and points of interest. Blue circles are source nodes, and red squares are targets. The defenders can use an expected number of 3 edges to set checkpoints, and the attackers always use the shortest path to attack the targets. If a checkpoint is in the path of the attacker, defenders get reward 1, 0 otherwise. The defenders' security strategy is to set a checkpoint on the highlighted red edges with equal probability $\frac{3}{7}$, and the game value is $\frac{3}{7}$. Notice that the highlighted edges are also the minimum number of edges that can cut all paths (only shortest paths are considered here) from sources to targets. }
  \label{fig: map}
\end{figure}
We particularize this model to the map of Figure \ref{fig: map} with $3$ source nodes indicated by blue circles and $4$ targets indicated by red squares. This map contains 36 nodes and 60 edges. We assume that the attackers use the shortest path to hit the target, and the defenders can use an expected number of $3$ edges. If a checkpoint is in the path of the attackers, defenders get reward $1$, $0$ otherwise. Let $x^{i}$ indicate the probability that a checkpoint is set on edge $i$, and $y^j$ indicate the probability that attackers choose path $j$. The urban security problem can be modeled as the following maxmin problem.
\begin{align*}
  V=&\max_{x\in \mathbb{R}^{60}}\min_{y\in \Delta(P)} x^TGy,\\
  s.t.& \mathbf{0}\leq x\leq \mathbf{1}\\
  & \mathbf{1}^T x =3
\end{align*}
where $P$ is the set of paths that attackers can take. Notice that $x^i$ is the probability that a checkpoint is set on edge $i$, but the vector $x$ is not a probability vector. The constraint $\mathbf{1}^T x =3$ corresponds to the requirement that the defenders can use an expected number of $3$ edges. According to the strong duality theorem, it can be transformed to an LP problem
\begin{align*}
  V=&\max_{x\in \mathbb{R}^{60}} \ell,\\
  s.t.& G^T x \geq \ell\mathbf{1} \\
  & \mathbf{0}\leq x\leq \mathbf{1}\\
  & \mathbf{1}^T x =3
\end{align*}
The corresponding canonical form is as follows.
\begin{align*}
 V=&\max_{x\in \mathbb{R}^{60}} [0\ldots0\ 1]x,\\
  s.t. &\left[\begin{array}{cc}
  -G^T T &\mathbf{1} \\
  \mathbf{1}^T & 0\\
  -I_{n-1} & \mathbf{0} \\
  I_{n-1} & \mathbf{0}
  \end{array}
  \right]x \leq \left[\begin{array}{c}
  3G^Te_n \\
  3\\
  \mathbf{0} \\
  \mathbf{1}
  \end{array}
  \right]
\end{align*}
We use regular shadow vertex method to solve this problem, and find that the security strategy is to set up checkpoints at red highlighted edges with probability $\frac{3}{7}$, and the value of the game is $\frac{3}{7}$. With $2.4$GHz CPU and $4$GB memory, the computation time is about $0.35$ seconds, which is comparable to the time reported in \cite{tsai2010urban} to compute an equilibrium strategy for a grid of similar size (35 nodes and 58 edges) representing south Mumbai. Notice that the highlighted edges are also the minimum number of edges that can cut all paths (only shortest paths are considered here) from sources to targets.

Now, assuming there is a new target in the map, we use the iterative shadow vertex method to update the security strategy. There are two small differences between this situation and that described earlier: 1) $x$ is not a probability vector, 2) each new target will result in several new paths, which adds several columns to the payoff matrix at a time. More precisely, equation (26) is changed to $[-g^TT\ \mathbf{1}]x \leq 3g^T e_n$, where $g$ is the newly added payoff columns. However, these differences only affect our theoretical investigation of the algorithm, and algorithm \ref{algorithm: iterative shadow vertex method} is in fact applicable as is to this scenario as well. We change the new target node from left to right, from bottom to top. The average computation time is $0.1429$ seconds. Compared with the computation time of regular shadow vertex method, which is about $0.35$ seconds, iterative shadow vertex method cuts more than half of the average computation time.

\section{Conclusion and future work}
This paper studies how to efficiently update the saddle-point strategy of one player in a matrix game when the other player can add new actions in the game. We provide an iterative shadow vertex method to solve this problem, and show that the computational complexity is strictly less than the regular shadow vertex method. Moreover, this paper also presents a necessary and sufficient condition that a new saddle-point strategy is needed, and analyzes the probability of re-computing the saddle point strategy. Our simulation results demonstrates the main results.

A direct extension of the problem in this paper is its dual problem, i.e. the case when player 1 has a growing action set. In this case, the corresponding LP has new variables whose dual problem is exactly the same problem as studied in this paper. We can use iterative shadow vertex method to solve its dual problem first, and then figure out the optimal solution from the optimal solution of its dual problem. A further extension is the case when both players have growing action sets. A proposed direction is to deal with player 2's new action first, and then deal with player 1's action set.

\bibliography{}
\end{document}